\documentclass[10pt, a4paper]{amsart}
\usepackage{amsmath, amsfonts, amssymb, mathtools, url, graphicx, graphicx}
\usepackage{newtxtext}
\usepackage{newtxmath}

\usepackage{enumitem}
\usepackage{algorithm,algorithmic}
\usepackage[usenames, dvipsnames]{color}
\allowdisplaybreaks
\usepackage{tikz}
\usetikzlibrary{automata,positioning}
\usepackage{caption}
\usepackage{subcaption}
\newtheorem{theorem}{Theorem}

\newtheorem{lemma}{Lemma}

\newtheorem{defn}{Definition}
\newtheorem{assump}{Assumption}

\newtheorem{prob}{Problem}
\newtheorem{rem}{Remark}

\newtheorem{experiment}{Experiment}

\newcommand{\norm}[1]{\left\lVert{#1}\right\rVert}
\newcommand{\abs}[1]{\left\lvert{#1}\right\rvert}

\newcommand{\pmat}[1]{\begin{pmatrix}#1\end{pmatrix}}

\renewcommand{\geq}{\geqslant}

\renewcommand{\leq}{\leqslant}

\newcommand{\R}{\mathbb{R}}
\newcommand{\N}{\mathbb{N}}

\newcommand{\is}{i_{s}}
\newcommand{\iu}{i_{u}}
\newcommand{\Svec}{\mathcal{S}}
\newcommand{\KL}{\mathcal{KL}}
\newcommand{\wv}{\overline{w}}
\newcommand{\we}{\underline{w}}
\renewcommand{\Re}{\text{Re}}
\newcommand{\Kinfty}{\mathcal{K}_{\infty}}
\renewcommand{\L}{\mathcal{L}}
\newcommand{\ls}{\lambda_{i_{s}}}
\newcommand{\lu}{\lambda_{i_{u}}}
\renewcommand{\l}{\ell}

\DeclareMathOperator{\minimize}{minimize}
\DeclareMathOperator{\sbjto}{subject\;to}

\allowdisplaybreaks

\title{A scheduling algorithm for networked control systems}
\author{Atreyee Kundu}%
\thanks{The author is with the Department of Electrical Engineering, Indian Institute of Science Bangalore, Bengaluru - 560012, Karnataka, India, E-mail: \texttt{atreyeek@iisc.ac.in}}

\keywords{}

\date{\today}

\begin{document}


	\begin{abstract}
      		This paper deals with the design of scheduling logics for Networked Control Systems (NCSs) whose shared communication networks have limited capacity. We assume that among \(N\) plants, only \(M\:(< N)\) plants can communicate with their controllers at any time instant. We present an algorithm to allocate the network to the plants periodically such that stability of each plant is preserved. The main apparatus for our analysis is a switched systems representation of the individual plants in an NCS. We rely on multiple Lyapunov-like functions and graph-theoretic arguments to design our scheduling logics. The set of results presented in this paper is a continuous-time counterpart of the results proposed in \cite{abc}. We present a set of numerical experiments to demonstrate the performance of our techniques.
	\end{abstract}
	
	\maketitle
	
\section{Introduction}
\label{s:intro}
    Networked Control Systems (NCSs) are spatially distributed control systems in which the communication between plants and their controllers occurs through shared communication networks. NCSs find wide applications in sensor networks, remote surgery, haptics collaboration over the internet, automated highway systems, unmanned aerial vehicles, etc. \cite{Hespanha2007}. While the use of shared communication networks in NCSs offers flexible architectures and reduced installation and maintenance costs, the exchange of information between the plants and their controllers often suffers from network induced limitations and uncertainties.

    In this paper we deal with NCSs whose communication networks have limited bandwidth. While NCSs applications typically involve a large number of plants, the bandwidth of the shared network is often limited. Examples of communication networks with limited bandwidth include wireless networks (an important component of smart home, smart transportation, smart city, remote surgery, platoons of autonomous vehicles, etc.) and underwater acoustic communication systems. The scenario in which the number of plants sharing a communication network is higher than the capacity of the network is called \emph{medium access constraint}. This constraint motivates a need for allocating the network to the plants in a manner so that good qualitative properties of each plant in the NCS are preserved. This task of efficient allocation of a shared communication network is commonly referred to as a \emph{scheduling problem} and the corresponding allocation scheme is called a \emph{scheduling logic}. We are interested in algorithmic design of scheduling logics for NCSs.

    Scheduling logics can be classified broadly into two categories: \emph{static} and \emph{dynamic}. In case of the former, a finite length allocation scheme of the network is determined offline and is applied eternally in a periodic manner, while in case of the latter, the allocation of the shared network is determined based on some information about the plant (e.g., states, outputs, access status of sensors and actuators, etc.). For NCSs with continuous-time linear plants, static scheduling logics that preserve stability of all plants are characterized using common Lyapunov functions in \cite{Hristu2001} and piecewise Lyapunov-like functions with average dwell time switching in \cite{Lin2005}. A more general case of co-designing a static scheduling logic and control action is addressed using combinatorial optimization with periodic control theory in \cite{Rehbinder2004} and Linear Matrix Inequalities (LMIs) optimization with average dwell time technique in \cite{Dai2010}. In the discrete-time setting, a blend of multiple Lyapunov-like functions and graph theory was employed to design stability preserving periodic scheduling logics recently in \cite{abc}. The authors of \cite{Zhang2006} characterize static switching logics that ensure reachability and observability of the plants under limited communication, and design an observer-based feedback controller for these logics. The corresponding techniques were later extended to the case of constant transmission delays \cite{Hristu2008} and Linear Quadratic Gaussian (LQG) control problem \cite{Hristu_Zhang2008}. Event-triggered scheduling logics that preserve stability of all plants under communication delays are proposed in \cite{Al-Areqi'15}. In \cite{Quevedo2014} the authors propose a mechanism to allocate network resources by finding optimal node that minimizes a certain cost function in every network time instant. The design of dynamic scheduling logics for stability of each plant under both communication uncertainties and computational limitations is studied in \cite{Saha2015}. In \cite{Gatsis2016} a class of distributed control-aware random network access logics for the sensors such that all control loops are stabilizable, is presented. A dynamic scheduling logic based on predictions of both control performance and channel quality at run-time, is proposed recently in \cite{Ma2019}.

    In this paper we consider an NCS consisting of multiple continuous-time linear plants whose feedback loops are closed through a shared communication network. A block diagram of such an NCS is shown in Figure \ref{fig:ncs}.
    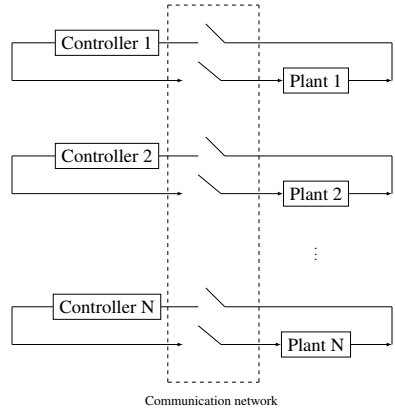
\begin{figure}[htbp]
    	\begin{center}
	\scalebox{0.5}{
	\begin{tikzpicture}[every path/.style={>=latex},base node/.style={draw,rectangle, scale = 1.4}]
	\node[base node] (a) at (-2,5) {Controller 1};
	\node[base node] (b) at (3.5,4) {Plant 1};
	\node[base node] (c) at (-2,2) {Controller 2};
	\node[base node] (d) at (3.5,1) {Plant 2};
	\node[base node] (e) at (-2,-2) {Controller N};
	\node[base node] (f) at (3.5,-3) {Plant N};	
	
	\draw (-4.5 ,5) edge (a);
	\draw (-4.5,5) edge (-4.5,4);
	\draw[->] (-4.5,4) -- (0,4);
	\draw (a) edge (0.4,5);
	\draw[->] (b) -- (5.5,4);
	\draw (5.5,4) edge (5.5,5);
	\draw[->] (1,4) -- (b);
	\draw[-.] (1,4) -- (0.4,4.5);
	\draw (5.5,5) edge (1.1,5);
	\draw[-.] (1.1,5) -- (0.6,5.5);

	\draw (-4.5 ,2) edge (c);
	\draw (-4.5,2) edge (-4.5,1);
	\draw[->] (-4.5,1) -- (0,1);
	\draw (c) edge (0.4,2);
	\draw[->] (d) -- (5.5,1);
	\draw (5.5,1) edge (5.5,2);
	\draw[->] (1,1) -- (d);
	\draw[-.] (1,1) -- (0.4,1.5);
	\draw (5.5,2) edge (1.1,2);
	\draw[-.] (1.1,2) -- (0.6,2.5);

	\draw (-4.5 ,-2) edge (e);
	\draw (-4.5,-2) edge (-4.5,-3);
	\draw[->] (-4.5,-3) -- (0,-3);
	\draw (e) edge (0.4,-2);
	\draw[->] (f) -- (5.5,-3);
	\draw (5.5,-3) edge (5.5,-2);
	\draw[->] (1,-3) -- (f);
	\draw[-.] (1,-3) -- (0.4,-2.5);
	\draw (5.5,-2) edge (1.1,-2);
	\draw[-.] (1.1,-2) -- (0.6,-1.5);

	\draw[dashed] (-0.4,-4) -- (-0.4,6);
	\draw[dashed] (2,-4) -- (2,6);
	\draw[dashed] (-0.4,-4) -- (2,-4);
	\draw[dashed] (-0.4,6) -- (2,6);

	\node (g) at (0.75,-4.5) {Communication network};
	\node (h) at (3.5,-0.5) {\(\vdots\)};

	\end{tikzpicture}
	}
	\caption{Block diagram of NCS}\label{fig:ncs}
	\end{center}
    \end{figure}
    We assume that the plants are unstable in open-loop and asymptotically stable in closed-loop. Due to a limited communication capacity of the network, only a few plants can exchange information with their controllers at any instant of time. Consequently, the remaining plants operate in open-loop at every time instant. We will design periodic scheduling logics that preserve global asymptotic stability (GAS) of each plant in the NCS. The set of results presented in this paper is a continuous-time counterpart of the results in \cite{abc}.

    We employ switched systems and graph theory as the main apparatuses for our analysis. We model the individual (open-loop unstable) plants of an NCS as switched systems, where the switching is between their open-loop (unstable mode) and closed-loop (stable mode) operations. In this setting no switched system can operate in stable mode for all time as that will destabilize some of the plants in the NCS. The search for a stabilizing scheduling logic then becomes the problem of finding switching logics that obey the limitations of the shared network and preserve stability. We assume that the exchange of information between a plant and its controller is not affected by communication uncertainties. We associate a weighted directed graph with the NCS that captures the communication limitation of the shared network, and design stabilizing switching logic{s} for each plant in the NCS. Multiple Lyapunov-like functions are employed for analyzing stability of the switched systems. Towards designing a stabilizing periodic scheduling logic, we combine stabilizing switching logics in terms of a class of cycles on the underlying weighted directed graph of the NCS that satisfies appropriate contractivity properties. We also discuss algorithmic construction of these cycles. It is known that periodic scheduling logics are easier to implement, often near optimal, and guarantee activation of each sensor and actuator, see \cite{Peters'16,Longo, Hristu'05} for detailed discussions. They are preferred for safety-critical control systems \cite[\S2.5.1]{Longo}. It is also observed in \cite{Orihuela'14, Peters'16} that periodic phenomenon appears in non-periodic schedules. We demonstrate the techniques proposed in this paper on a numerical example.

    The remainder of this paper is organized as follows: in \S\ref{s:prob_stat} we formulate the problem under consideration. The apparatuses for our design of scheduling logics and analysis of stability are described in \S\ref{s:prelims}. Our results appear in \S\ref{s:mainres}. We present a numerical example in \S\ref{s:num_ex} and conclude in \S\ref{s:concln} with a brief discussion of future research directions.

    We employ standard notations throughout the paper. \(\N\) is the set of natural numbers, \(\N_{0} = \N\cup\{0\}\), and \(\R\) is the set of real numbers. \(\norm{\cdot}\) denotes the Euclidean norm (resp., induced matrix norm) of a vector (resp., a matrix). For a finite set \(A\), we employ \(\abs{A}\) to denote its cardinality, i.e., the number of elements in \(A\). For a matrix \(P\in\R^{d\times d}\), \(\lambda_{\max}(P)\) and \(\lambda_{\min}(P)\) denote the maximum and minimum eigenvalues of \(P\), respectively. For a scalar \(a\), possibly complex, \(\Re(a)\) denotes the real part of \(a\).
\section{Problem statement}
\label{s:prob_stat}
    We consider an NCS with \(N\) plants whose dynamics are given by
    \begin{align}
    \label{e:plants}
        \dot{x}_i(t) = A_i x_i(t) + B_i u_i(t),\:\:x_i(0)=x_i^{0},\:\:t\in[0,+\infty[,
    \end{align}
    where \(x_i(t)\in\R^{d}\) and \(u_i(t)\in\R^{m}\) are the vectors of states and inputs of the \(i\)-th plant at time \(t\), respectively, \(i=1,2,\ldots,N\). All plants communicate with their remotely located state-feedback controllers
    \begin{align}
    \label{e:controllers}
        u_i(t) = K_i x_i(t),\:\:i=1,2,\ldots,N
    \end{align}
    through a shared communication network. The matrices \(A_i\in\R^{d\times d}\), \(B_i\in\R^{d\times m}\) and \(K_i\in\R^{m\times d}\), \(i=1,2,\ldots,N\) are known.

    We will operate under the following set of assumptions:
    \begin{assump}
    \label{a:capacity}
    \rm{
        The shared communication network has a limited communication capacity in the sense that at any time instant, only \(M\) plants \((0<M<N)\) can access the network. Consequently, \(N-M\) plants operate in open-loop at every time instant.
        }
    \end{assump}
    \begin{assump}
    \label{a:stability}
    \rm{
        The open-loop dynamics of each plant is unstable and each controller is stabilizing. More specifically, the matrices \(A_i+B_iK_i\), \(i=1,2,\ldots,N\) are stable (Hurwitz) and the matrices \(A_i\), \(i=1,2,\ldots,N\) are unstable.\footnote{Recall that a matrix \(A\in\R^{d\times d}\) is Hurwitz if every eigenvalue of \(A\) has strictly negative real part. We call \(A\) unstable if it is not Hurwitz.}
        }
    \end{assump}
    \begin{assump}
    \label{a:ideal_comm}
    \rm{
        The shared communication network is ideal in the sense that the exchange of information between the plants and their controllers is not affected by communication uncertainties.
        }
    \end{assump}

    In view of Assumptions \ref{a:capacity} and \ref{a:stability}, each plant in \eqref{e:plants} operates in two modes: (a) stable mode when the plant has access to the shared communication network and (b) unstable mode when the plant does not have access to the network. Let us denote the stable and unstable modes of the \(i\)-th plant as \(\is\) and \(\iu\), respectively, \(A_{\is} = A_i + B_i K_i\) and \(A_{\iu} = A_i\), \(i=1,2,\ldots,N\). In this paper we are interested in finding a mechanism to allocate the shared communication network to the plants such that stability of each plant in \eqref{e:plants} is preserved.

    Let
    \[
        \Svec := \{s\in\{1,2,\ldots,N\}^{M}\:|\:\text{all elements of \(s\) are distinct}\}
    \]
    be the set of vectors that consist of \(M\) distinct elements from \(\{1,2,\ldots,N\}\). We call a function \(\gamma:[0,+\infty[\to\Svec\) a \emph{scheduling logic}. There exists a diverging sequence of times \(0=:\tau_0<\tau_1<\tau_2<\cdots\) and a sequence of indices \(s_0, s_1, s_2,\ldots\) with \(s_j\in\Svec\), \(j=0,1,2,\ldots\) such that \(\gamma(t) = s_j\) for all \(t\in[\tau_j:\tau_{j+1}[\), \(j=0,1,2,\ldots\). In other words, \(\gamma\) specifies, at every time \(t\), \(M\) plants of the NCS which access the communication network at that time. The remaining \(N-M\) plants operate in open-loop, i.e., with \(u_i(t) = 0\).

    \begin{defn}{\cite[Lemma 4.4]{Khalil}}
    \label{d:gas}
    \rm{
        The \(i\)-th plant in \eqref{e:plants} is globally asymptotically stable (GAS) for a given scheduling logic \(\gamma\), if there exists a class \(\KL\) function \(\beta_i\) such that the following inequality holds:
        \begin{align}
        \label{e:gas}
            \norm{x_i(t)}\leq \beta_i(\norm{x_i(0)},t)\:\:\text{for all}\:x_i(0)\in\R^{d}\:\:\text{and}\:t\in[0,+\infty[.
        \end{align}
    }
    \end{defn}
    We will solve the following problem:
    \begin{prob}
    \label{prob:main}
        Given the matrices \(A_i\), \(B_i\), \(K_i\), \(i=1,2,\ldots,N\) and a number \(M (< N)\), design a scheduling logic \(\gamma\) that preserves GAS of each plant in \eqref{e:plants}.
    \end{prob}
    We will model each plant in \eqref{e:plants} as a switched system and associate a labelled and weighted directed graph with the NCS under consideration. Scheduling logics \(\gamma\) that preserve GAS of each plant \(i\) in \eqref{e:plants} will be designed by concatenating cycles on this directed graph that satisfy certain contractivity properties. In the sequel we will refer to such scheduling logics as stabilizing scheduling logics. Prior to presenting our solution to Problem \ref{prob:main}, we catalog a set of preliminaries required for our analysis.
\section{Preliminaries}
\label{s:prelims}
\subsection{Individual plants and switched systems}
\label{ss:swsys-plants}
    We model the dynamics of the \(i\)-th plant in \eqref{e:plants} as a switched system
    \begin{align}
    \label{e:i-swsys}
        \dot{x}_i(t) = A_{\sigma_i(t)}x_i(t),\:\:x_i(0) = x_i^{0},\:\:\sigma_i(t)\in\{\is,\iu\},\:\:t\in[0,+\infty[,
    \end{align}
    where the subsystems are \(\{A_{\is},A_{\iu}\}\) and a switching logic \(\sigma_i:\N_0\to\{\is,\iu\}\) satisfies:
    \begin{align*}
        \sigma_i(t) =
        \begin{cases}
            \is,\:\:\text{if}\:i\:\text{is an element of}\:\gamma(t),\\
            \iu,\:\:\text{otherwise}.
        \end{cases}
    \end{align*}
    Clearly, a switching logic \(\sigma_i\), \(i=1,2,\ldots,N\) is a function of the scheduling logic \(\gamma\). In order to ensure GAS of the individual plants, it therefore, suffices to design a \(\gamma\) that renders each \(\sigma_i\) stabilizing in the following sense: \(\sigma_i\) guarantees GAS of the switched system \eqref{e:i-swsys} for each \(i=1,2,\ldots,N\).

    \begin{lemma}
    \label{fact:key1}
    \rm{
        For each \(i=1,2,\ldots,N\), there exist pairs \((P_p,\lambda_p)\), \(p\in\{\is,\iu\}\), where \(P_p\in\R^{d\times d}\) are symmetric and positive definite matrices, and \(\lambda_{\is} > 0\), \(\lambda_{\iu}\leq 0\), such that with
        \begin{align}
        \label{e:Vp_defn}
            \R^{d}\ni\xi\mapsto V_{p}(\xi) : = \xi^\top P_p \xi\in[0,+\infty[
        \end{align}
        we have
        \begin{align}
        \label{e:key_ineq1}
            V_{p}(z_{p}(t))\leq\exp(-\lambda_p t)V_{p}(z_{p}(0)),\:\:t\in[0,+\infty[,
        \end{align}
        and \(z_p(\cdot)\) solves the \(p\)-th recursion in \eqref{e:i-swsys}, \(p\in\{\is,\iu\}\).
    }
    \end{lemma}
    \begin{proof}
        Fix \(i\in\{1,2,\ldots,N\}\). Since \(A_{\is}\) is Hurwitz, the matrix \(P_{\is}\) can be selected as a symmetric and positive definite solution to the Lyapunov equation
        \[
            A_{\is}^\top P_{\is} + P_{\is}A_{\is} + Q_{\is} = 0_{d}
        \]
        for some pre-selected symmetric and positive definite matrix \(Q_{\is}\in\R^{d\times d}\) \cite[Corollary 11.9.1]{Bernstein}. A straightforward calculation gives
        \[
            \frac{d}{dt}\biggl(V_{\is}\bigl(z_{\is}(t)\bigr)\biggr) = -z_{\is}(t)^\top Q_{\is} z_{\is}(t),
        \]
        where \(z_{\is}(\cdot)\) solves the \(\is\)-th system dynamics in \eqref{e:i-swsys}. Recall that for any symmetric matrix \(Z\in\R^{d\times d}\), we have \cite[Lemma 8.4.3]{Bernstein}
        \begin{align}
        \label{e:halp_ineq}
            \lambda_{\min}(Z)\norm{z}^{2} \leq z^\top Z z \leq \lambda_{\max}(Z)\norm{z}^{2}\:\:\text{for all}\:z\in\R^{d}.
        \end{align}
        It, therefore, follows that
        \begin{align*}
            -z_{\is}(t)^\top Q_{\is}z_{\is}(t)\leq-\frac{\lambda_{\min}(Q_{\is})}{\lambda_{\max}(P_{\is})}z_{\is}(t)^\top P_{\is}z_{\is}(t).
        \end{align*}
        Defining \(\displaystyle{\lambda_{\is} = \frac{\lambda_{\min}(Q_{\is})}{\lambda_{\max}(P_{\is})}}\), we arrive at
        \begin{align*}
            \frac{d}{dt}\biggl(V_{\is}\bigl(z_{\is}(t)\bigr)\biggr)\leq-\lambda_{\is}V_{\is}\bigl(z_{\is}(t)\bigr),
        \end{align*}
        which gives \eqref{e:key_ineq1} with \(\lambda_{\is} > 0\).

        Now, with \(A_{\iu}\) unstable, there exists \(\varepsilon_{\iu} > 0\) such that \(A_{\iu}-\varepsilon_{\iu}I_{d}\) is Hurwitz. Fix a symmetric and positive definite matrix \(Q_{\iu}\), and let \(P_{\iu}\) be the symmetric and positive definite solution to the Lyapunov equation
        \[
            (A_{\iu}-\varepsilon_{\iu}I_d)^\top P_{\iu} + P_{\iu}(A_{\iu}-\varepsilon_{\iu}I_d) + Q_{\iu} = 0_d.
        \]
        It follows that
        \[
            \frac{d}{dt}\biggl(V_{\iu}\bigl(z_{\iu}(t)\bigr)\biggr) = -z_{\iu}(t)^\top Q_{\iu}z_{\iu}(t),
        \]
        where \(z_{\iu}(\cdot)\) solves the \(\iu\)-th system dynamics in \eqref{e:i-swsys}. Applying \eqref{e:halp_ineq}, we arrive at
        \[
            -z_{\iu}(t)^\top Q_{\iu}z_{\iu}(t) \leq -\biggl(2\varepsilon_{\iu}-\frac{\lambda_{\min}(Q_{\iu})}{\lambda_{\max}(P_{\iu})}\biggr)z_{\iu}(t)^\top P_{\iu}z_{\iu}(t).
        \]
        Defining \(\displaystyle{\lambda_{\iu} = 2\varepsilon_{\iu}-\frac{\lambda_{\min}(Q_{\iu})}{\lambda_{\max}(P_{\iu})}}\), we arrive at
        \[
            \frac{d}{dt}\biggl(V_{\iu}\bigl(z_{\iu}(t)\bigr)\biggr)\leq-\lambda_{\iu}V_i\bigl(z_{\iu}(t)\bigr).
        \]
        Notice that any scalar larger than \(\Re\bigl(\lambda_{\max}(A_{\iu})\bigr)\) is a valid choice of \(\varepsilon_{\iu}\). One, therefore, needs to choose \(\varepsilon_{\iu}\) and \(Q_{\iu}\) such that \(\displaystyle{2\varepsilon_{\iu}-\frac{\lambda_{\min}(Q_{\iu})}{\lambda_{\max}(P_{\iu})}\leq 0}\).

        This completes our proof of Fact \ref{fact:key1}.
    \end{proof}
    \begin{lemma}
    \label{fact:key2}
    \rm{
        For each \(i=1,2,\ldots,N\), there exist \(\mu_{pq}\geq 1\) such that
        \begin{align}
        \label{e:key_ineq2}
            V_{q}(\xi)\leq\mu_{pq}V_p(\xi)\:\:\text{for all}\:\xi\in\R^{d}\:\:\text{and}\:\:p,q,\in\{\is,\iu\}.
        \end{align}
        }
    \end{lemma}
    \begin{proof}
        Linear comparability of \(V_p\)'s is clear from the definition of \(V_p\), \(p\in\{\is,\iu\}\) in \eqref{e:Vp_defn}. The assertion of Fact \ref{fact:key2} follows at once.
    \end{proof}
    The function \(V_{p}\), \(p\in\{\is,\iu\}\), \(i=1,2,\ldots,N\) are called Lyapunov-like functions. The scalars \(\lambda_p\), \(p\in\{\is,\iu\}\) give quantitative measures of (in)stability associated to (un)stable modes of operation of the \(i\)-th plant. A \emph{tight} estimate of the scalars \(\mu_{pq}\), \(p,q\in\{\is,\iu\}\) is given by \(\lambda_{\max}(P_q P_p^{-1})\) \cite[Proposition 4]{chatterjee'14}. Facts \ref{fact:key1} and \ref{fact:key2} will be employed in our design of \(\gamma\).
\subsection{NCS and directed graphs}
\label{ss:NCS+digraph}
    We associate a directed graph \(G(V,E)\) with the NCS under consideration. The vertex set \(V\) contains \(N\choose M\) vertices that are labelled distinctly. The label associated to a vertex \(v\in V\) is given by \(L(v) = \{\ell_v(1),\ell_v(2),\ldots,\ell_v(N)\}\), where \(\ell_{v}(i) = \is\) for any \(M\) elements of \(\{1,2,\ldots,N\}\) and \(\ell_v(i) = \iu\) for the remaining \(N-M\) elements. The edge set \(E\) contains directed edges \((u,v)\) for every \(u,v\in V\), \(u\neq v\).

    Let the functions \(\wv(v) = \pmat{\wv_1(v)\\\wv_2(v)\\\vdots\\\wv_{N}(v)}\), \(v\in V\) with
    \begin{align}
    \label{e:vertex_wt}
        \wv_i(v) =
        \begin{cases}
            -\abs{\lambda_{\is}},\:\:\text{if}\:\ell_{v}(i) = \is,\\
            \abs{\lambda_{\iu}},\:\:\text{if}\:\ell_{v}(i) = \iu,
        \end{cases}
        i=1,2,\ldots,N,
    \end{align}
    and \(\we(u,v) = \pmat{\we_1(u,v)\\\we_2(u,v)\\\vdots\\\we_{N}(u,v)}\), \((u,v)\in E\) with
    \begin{align}
    \label{e:edge_wt}
        \we_i(u,v) =
        \begin{cases}
            \ln\mu_{\is\iu},\:\text{if}\:\ell_u(i) = \is\:\text{and}\:\ell_v(i) = \iu,\\
            \ln\mu_{\iu\is},\:\text{if}\:\ell_{u}(i) = \iu\:\text{and}\:\ell_v(i) = \is,\\
            0,\:\:\text{otherwise},
        \end{cases}
        \hspace*{-0.4cm}i=1,\ldots,N,
    \end{align}
     be the weight associated to a vertex \(v\in V\) and the weight associated to an edge \((u,v)\in E\), respectively. Here, the scalars \(\lambda_p\), \(p\in\{\is,\iu\}\) and \(\mu_{pq}\), \(p,q\in\{\is,\iu\}\), \(i=1,2,\ldots,N\), are as described in Fact \ref{fact:key1} and Fact \ref{fact:key2}, respectively.
    \begin{rem}
    \label{rem:compa1}
    \rm{
    The association of a directed graph with an NCS was first proposed in \cite{abc}. Here we employ a natural extension of this association to the continuous-time setting. The label \(L(v)\) of a vertex \(v\in V\) gives a combination of \(M\) plants operating in stable mode and the remaining \(N-M\) plants operating in unstable mode. Since \(V\) contains \(N\choose M\) vertices and the label of each vertex is distinct, it follows that the set of vertex labels consists of all possible combinations of \(M\) plants accessing the communication network and \(N-M\) plants operating in open-loop. A directed edge \((u,v)\) from a vertex \(u\) to a vertex \(v\:(\neq u)\) corresponds to a transition from a set of \(M\) plants accessing the communication network (as specified by \(L(v)\)). The vertex (subsystem) weights of \(G(V,E)\) capture the rate of increase/decrease of the Lyapunov-like functions \(V_p\), \(p\in\{\is,\iu\}\) and the edge weights of \(G(V,E)\) capture the ``jump'' between Lyapunov-like functions \(V_p\) and \(V_{q}\), \(p,q\in\{\is,\iu\}\). This choice of weights is employed with the objective to compensate the increase in \(V_p\), \(p\in\{\is,\iu\}\) caused by activation of unstable mode \(\iu\) and switches between stable and unstable modes (\(\is\) to \(\iu\) and \(\iu\) to \(\is\)) by the decrease in \(V_p\), \(p\in\{\is,\iu\}\) caused by the activation of the stable modes \(\is\), \(i=1,2,\ldots,N\), as will be useful in our analysis for GAS of the individual plants.
    }
    \end{rem}

    Recall that \cite[p.\ 4]{Bollobas} a \emph{cycle} on the directed graph \(G(V,E)\) is an alternating (finite) sequence of vertices and edges that begin and end on the same vertex, e.g., \(W = \tilde{v}_0,(\tilde{v}_0,\tilde{v}_1),\tilde{v}_1\),\\\(\ldots,\tilde{v}_{n},(\tilde{v}_{n},\tilde{v}_{0}),\tilde{v}_{0}\). The number of edges that appear in the sequence is called the \emph{length} of the cycle. Here the length of \(W\) is \(n\). We will employ the following class of cycles in our design of stabilizing scheduling logics:
    \begin{defn}
    \label{d:contrac_cycle}
    \rm{
        A cycle \(W = v_0,(v_0,v_1),v_1,\ldots,v_{n-1},(v_{n-1}\),\(v_0),v_0\) on \(G(V,E)\) is called \emph{\(T\)-contractive} if there exist \(\R\ni T_{v_{j}} > 0\), \(j=0,1,\ldots,n-1\), \(2\leq n\leq \abs{V}\), such that the following set of inequalities is satisfied:
        \begin{align}
        \label{e:contrac}
            \Xi_i(W) := \sum_{j=0}^{n-1}\wv_i(v_j)T_{v_{j}} + \sum_{\substack{{j=0}\\{v_{n}:=v_0}}}^{n-1}\we_i(v_j,v_{j+1}) < 0
        \end{align}
        for all \(i=1,2,\ldots,N\), where \(n\) is the length of \(W\), \(\wv(v_{j})\) is the weight associated to vertex \(v_j\), \(\wv_i(v_j)\) is the \(i\)-th element of \(\wv(v_j)\), and \(\we(v_j,v_{j+1})\) is the weight associated to the edge \((v_j,v_{j+1})\), \(\we_i(v_j,v_{j+1})\) is the \(i\)-th element of \(\we(v_{j},v_{j+1})\), \(i=1,2,\ldots,N\), \(j=0,1,\ldots,n-1\).
    }
    \end{defn}
    \begin{rem}
    \label{rem:compa2}
    \rm{
    Definition \ref{d:contrac_cycle} is a natural extension of \cite[Definition 2]{abc} to the continuous-time setting. The scalars \(T_{v_{j}}\), \(j=0,1,\ldots,n-1\) will be employed to associate a time duration with every vertex \(v_j\), \(j=0,1,\ldots,n-1\) that appears in \(W\). This time duration will determine how long a set of \(M\) plants can access the shared communication network while preserving GAS of all plants in the NCS under consideration. Naturally, \(T_{v_{j}}\)'s are real numbers for the continuous-time case, while they are integers for the discrete-time setting. In the sequel we will call the scalars \(T_{v_{j}}\), \(j=0,1,\ldots,n-1\) as the \(T\)-factor of vertex \(v_j\), \(j=0,1,\ldots,n-1\).
    }
    \end{rem}
    \begin{rem}
    \label{rem:contrac_cycle}
    \rm{
        The concept of contractive cycles and its variants have appeared in the context of designing switching signals that preserve stability of continuous-time switched systems earlier in the literature, see e.g., \cite{xyz}. In this paper we will use the notion of \(T\)-contractive cycles to address a harder problem of simultaneously preserving GAS of \(N\) switched systems.
    }
    \end{rem}
    We now move on to our solution to Problem \ref{prob:main}.
\section{Results}
\label{s:mainres}
    We will solve Problem \ref{prob:main} in two steps:
    \begin{itemize}[label = \(\circ\), leftmargin = *]
        \item First, we present an algorithm that constructs periodic scheduling logics \(\gamma\) by employing a \(T\)-contractive cycle on the underlying directed graph \(G(V,E)\) of an NCS and its corresponding \(T\)-factors.
        \item Second, we show that scheduling logics obtained from our algorithm preserve GAS of each plant in \eqref{e:plants}.
    \end{itemize}
    We will also present an algorithm to design \(T\)-contractive cycles on \(G(V,E)\).

    \begin{algorithm}
			\caption{Construction of periodic scheduling logics} \label{algo:sched_logic}
		\begin{algorithmic}[1]
			\renewcommand{\algorithmicrequire}{\textbf{Input:}}
			\renewcommand{\algorithmicensure}{\textbf{Output:}}
			
			\REQUIRE a \(T\)-contractive cycle \(W = v_0,(v_0,v_1),v_1,\ldots,v_{n-1}\),\((v_{n-1},v_0),v_0\) and the corresponding \(T\)-factors \(T_0\), \(T_1,\ldots,T_{n-1}\).
			\ENSURE a periodic scheduling logic \(\gamma\)
			
			 \hspace*{-0.6cm}\textit {Step I: For each vertex \(v_{j}\), \(j = 0,1,\ldots,n-1\), pick the elements
			 \(i\) with label \(\l_{v_{j}}(i) = \is\), \(i=1,2,\ldots,N\), and construct \(M\)-
			 dimensional vectors \(s_{j}\), \(j = 0,1\),\(\ldots,n-1\).}
			 \FOR{ \(j = 0,1,\ldots,n-1\)}
			 	\STATE Set \(p = 0\).
				\FOR {\(i = 1,2,\ldots,N\)}
					\IF {\(\l_{v_{j}}(i) = \is\)}
					\STATE Set \(p=p+1\) and \(u_{j}(p) = i\).
					\ENDIF
				\ENDFOR
			 \ENDFOR
			
			 \hspace*{-0.6cm}\textit {Step II: Construct a scheduling logic using the vectors \(s_{j}\), \(j = 0,1,\ldots,n-1\) obtained in Step I and the \(T\)-factors \(T_{v_{j}}\),\(j = 0,1,\ldots,n-1\)}
			 \STATE Set \(p=0\) and \(\tau_{0} = 0\)
				 \FOR {\(q = pn, pn+1,\ldots, (p+1)n-1\)}			\label{algo_step:rec}
			 	\STATE Set \(\gamma(\tau_{q}) = s_{q-pn}\) and \(\tau_{q+1} = \tau_{q} + T_{v_{q-pn}}\).
				\STATE Output \(\tau_{q}\) and \(\gamma(\tau_{q})\).
			 \ENDFOR
			 \STATE Set \(p = p+1\) and go to \ref{algo_step:rec}.
		\end{algorithmic}
	\end{algorithm}
    Given the matrices \(A_i\), \(B_i\), \(K_i\), \(i=1,2,\ldots,N\) and a number \(M\), Algorithm \ref{algo:sched_logic} designs a scheduling logic \(\gamma\), that specifies, at every time, \(M\) plants that access the shared communication network at that time. Algorithm \ref{algo:sched_logic} is a continuous-time counterpart of \cite[Algorithm 1]{abc}. The key ingredient of Algorithm \ref{algo:sched_logic} is a \(T\)-contractive cycle on the underlying directed graph \(G(V,E)\) of the NCS under consideration. In Step I, corresponding to each vertex \(v_j\), \(j=0,1,\ldots,n-1\), a vector \(s_j\), \(j=0,1,\ldots,n-1\) of length \(M\) is created with the elements \(i\in\{1,2,\ldots,N\}\) for which \(\ell_{v_{j}}(i) = \is\). In Step II, a scheduling logic \(\gamma\) is constructed from the vectors \(s_j\), \(j=0,1,\ldots,n-1\) and the corresponding \(T\)-factors \(T_{v_{j}}\), \(j=0,1,\ldots,n-1\). Sets of \(M\) plants corresponding to the elements in \(s_j\) access the shared communication network for \(T_{v_{j}}\) duration of time, \(j=0,1,\ldots,n-1\), and the process is repeated. Clearly, a scheduling logic \(\gamma\) constructed as above, is periodic with period \(\displaystyle{\sum_{j=0}^{n-1}T_{v_{j}}}\). Theorem \ref{t:mainres} asserts that a scheduling logic obtained from Algorithm \ref{algo:sched_logic} is stabilizing.
    \begin{theorem}
    \label{t:mainres}
        Consider an NCS described in \S\ref{s:prob_stat}. Let the matrices \(A_i\), \(B_i\), \(K_i\), \(i=1,2,\ldots,N\) and a number \(M\:(<N)\) be given. Then each plant \(i\) in \eqref{e:plants} is GAS under a scheduling logic \(\gamma\) obtained from Algorithm \ref{algo:sched_logic}.
    \end{theorem}
    \begin{proof}{(Sketch)}
        Consider the NCS described in \S\ref{s:prob_stat} and its underlying directed graph \(G(V,E)\). Let \(W = v_0,(v_0,v_1),v_1,\ldots,v_{n-1}\),\((v_{n-1},v_0),v_0\) be a \(T\)-contractive cycle on \(G(V,E)\). Consider a scheduling logic \(\gamma\) obtained from Algorithm \ref{algo:sched_logic} constructed by employing \(W\). We will show that \(\gamma\) preserves GAS of each plant in \eqref{e:plants}.

        Fix an \(i\in\{1,2,\ldots,N\}\) arbitrary. It suffices to show that the switched system \eqref{e:i-swsys} is GAS under the switching signal \(\sigma_i\) corresponding to \(\gamma\).

        Fix a time \(t > 0\). Recall that \(0=:\tau_0<\tau_1<\cdots\) are the points in time where \(\gamma\) changes values. Let \(N_{t}^{\gamma}\) denote the total number of times \(\gamma\) has changed its values on \(]0,t]\).

        In view of \eqref{e:key_ineq1}, we have
        \begin{align}
        \label{e:pf1_step1}
            V_{\sigma_i(t)}(x_i(t))\leq \exp\biggl(-\lambda_{\sigma_i(\tau_{N_{t}^{\gamma}})}\bigl(t-\tau_{N_{t}^{\gamma}}\bigr)\biggr)
            V_{\sigma_i(t)}(x_i(\tau_{N_{t}^{\gamma}})).
        \end{align}
        A straightforward iteration of \eqref{e:pf1_step1} applying \eqref{e:key_ineq1} and \eqref{e:key_ineq2} leads to
        \begin{align}
        \label{e:pf1_step2}
            V_{\sigma_i(t)}(x_i(t))&\leq \exp\biggl(-\sum_{\substack{{j=0}\\{\tau_{N_{t}^{\gamma}+1}:=t}}}^{N_{t}^{\gamma}}\lambda_{\sigma_i(\tau_{j})}(\tau_{j+1}-\tau_{j})\biggr)
            \times\prod_{j=0}^{N_{t}^{\gamma}-1}\mu_{\sigma_i(\tau_{j})\sigma_i(\tau_{j+1})} V_{\sigma_i(0)}(x_i(0)).
        \end{align}

        Now,
        \begin{align*}
            \prod_{j=0}^{N_{t}^{\gamma}-1}\mu_{\sigma_i(\tau_{j})\sigma_i(\tau_{j+1})} &= \exp\ln\Biggl(\prod_{j=0}^{N_{t}^{\gamma}-1}\mu_{\sigma_i(\tau_{j})\sigma_i(\tau_{j+1})}\Biggr)\\
            &= \exp\Biggl(\sum_{j=0}^{N_{t}^{\gamma}-1}\ln\mu_{\sigma_i(\tau_j)\sigma_i(\tau_{j+1})}\Biggr)\\
            &= \exp\Biggl(\sum_{p\in\{\is,\iu\}}\sum_{j=0}^{N_{t}^{\gamma}-1}\sum_{\substack{{p\to q:}\\{q\in\{\is,\iu\},}\\{\sigma_i(\tau_j)=p,}\\{\sigma_i(\tau_{j+1})=q}}}\ln\mu_{pq}\Biggr).
        \end{align*}
        Let \(N_{pq}(s,t)\) denote the total number of transitions from subsystem (mode) \(p\) to subsystem (mode) \(q\), \(p,q\in\{\is,\iu\}\) on \(]s,t]\). Then the RHS of the above expression is equal to
        \begin{align}
        \label{e:pf1_step3}
            \exp\biggl(\ln\mu_{\is\iu}N_{\is\iu}(0,t)+\ln\mu_{\iu\is}N_{\iu\is}(0,t)\biggr),
        \end{align}
        since \(\ln\mu_{\is\is}=\ln\mu_{\iu\iu} = 0\). Moreover,
        \begin{align*}
            &\exp\Biggl(-\sum_{\substack{{j=0}\\{\tau_{N_{t}^{\gamma}+1}:=t}}}^{N_{t}^{\gamma}}\lambda_{\sigma_i(\tau_{j})}(\tau_{j+1}-\tau_{j})\Biggr)\\
            =&\exp\Biggl(-\sum_{\substack{{j=0}\\{\tau_{N_t^{\sigma}+1}:=t}}}^{N_{t}^{\gamma}}\biggl(\sum_{p\in\{\is,\iu\}}\mathrm{1}(\sigma_i(\tau_j)=p)\lambda_p(\tau_{j+1}-\tau_{j})\biggr)\Biggr).
        \end{align*}
        Let \(D_s(s,t)\) and \(D_u(s,t)\) denote the total durations of activation of the stable and unstable modes of \(i\) on \(]s,t]\), respectively. Consequently, the above expression is equal to
        \begin{align}
        \label{e:pf1_step4}
            \exp\biggl(-\abs{\lambda_{\is}}D_s(0,t) + \abs{\lambda_{\iu}}D_u(0,t)\biggr).
        \end{align}
        Substituting \eqref{e:pf1_step3} and \eqref{e:pf1_step4} in \eqref{e:pf1_step2}, we obtain
        \begin{align}
        \label{e:pf1_step5}
            V_{\sigma_i(t)}(x_i(t))\leq\psi_i(t)V_{\sigma_i(0)}(x_i(0)),
        \end{align}
        where
        \begin{align}
        \label{e:pf1_step6}
            \N\ni t\mapsto \psi_i(t) &:= \exp\biggl(-\abs{\lambda_{\is}}D_s(0,t)+\abs{\lambda_{\iu}}D_u(0,t)\nonumber\\
            &\:\:+\ln\mu_{\is\iu}N_{\is\iu}(0,t)+\ln\mu_{\iu\is}N_{\iu\is}(0,t)\biggr).
        \end{align}

        From the definition of \(V_p\), \(p\in\{\is,\iu\}\) in \eqref{e:Vp_defn} and properties of positive definite matrices \cite[Lemma 8.4.3]{Bernstein}, it follows that
        \begin{align}
        \label{e:pf1_step7}
            \norm{x_i(t)}\leq c\psi_i(t)\norm(x_i(0))\:\:\text{for all}\:t\in[0,+\infty[,
        \end{align}
        where \(\displaystyle{c = \sqrt{\frac{\max_{p\in\{\is,\iu\}}\lambda_{\max}(P_p)}{\min_{p\in\{\is,\iu\}}\lambda_{\min}(P_p)}}}\). In order to establish GAS of \eqref{e:i-swsys}, we require to show that \(c\norm{x_i(0)}\psi_i(t)\) can be bounded above by a class \(\KL\) function. It is immediate that \(c\norm{x_i(0)}\) is a class \(\Kinfty\) function. It remains to show that \(\psi_i(t)\) is bounded above by a function belonging to class \(\L\).

        Recall that \(\gamma\) is constructed by employing a \(T\)-contractive cycle \(W = v_{0},(v_{0},v_{1}),v_{1},\ldots\),\\\(v_{n-1},(v_{n-1},v_{0}),v_{0}\) on \(G\), and \(T_{v_{j}}\), \(j = 0,1,\ldots,n-1\) are the \(T\)-factors associated to vertices \(v_{j}\), \(j = 0,1,\ldots,n-1\). Let \(\displaystyle{T_{W} := \sum_{j=0}^{n-1}T_{v_{j}}}\), \(t\geq mT_{W}\), \(m\in\N_{0}\), and \(\Xi_{i}(W) = -\varepsilon_{i}\), \(\varepsilon_{i} > 0\), where \(\Xi_{i}(W)\) is as defined in \eqref{e:contrac}.
		By construction of \(\gamma\), we have
		\begin{align}
		\label{e:pf1step7}
			\psi_{i}(t) &= \exp\Biggl(-\abs{\ls}D_{s}(0,t)+\abs{\lu}D_{u}(0,t)
+\ln\mu_{\is\iu}N_{\is\iu}(0,t)+\ln\mu_{\iu\is}N_{\iu\is}(0,t)\Biggr)\nonumber\\
			&=-\abs{\ls}D_{s}(0,mT_{W})-\abs{\ls}D_{s}(mT_{W},t)+\abs{\lu}D_{u}(0,mT_{W})+\abs{\lu}D_{u}(mT_{W},t)\nonumber\\
			&\:\:\:\:\:+\ln\mu_{\is\iu}N_{\is\iu}(0,mT_{W})+\ln\mu_{\is\iu}N_{\is\iu}(mT_{W},t)\nonumber\\
			&\:\:\:\:\:+\ln\mu_{\iu\is}N_{\iu\is}(0,mT_{W})+\ln\mu_{\iu\is}N_{\iu\is}(mT_{W},t).
		\end{align}
		Notice that
		\begin{align*}
			&-\abs{\ls}D_{s}(0,mT_{W}) + \abs{\lu}D_{u}(0,mT_{W})+ \ln\mu_{\is\iu}N_{\is\iu}(0,mT_{W}) + \ln\mu_{\iu\is}N_{\iu\is}(0,mT_{W})\\
			&=-\abs{\ls}m\sum_{\substack{{j:\l_{v_{j}}(i) = \is}\\{j = 0,1,\ldots,n-1}}}T_{v_{j}} + \abs{\lu}m\sum_{\substack{{j:\l_{v_{j}}(i) = \iu}\\{j = 0,1,\ldots,n-1}}}T_{v_{j}}\\
			&\hspace*{1cm}+\ln\mu_{\is\iu}m\#(\is\to\iu)_{W}+\ln\mu_{\iu\is}m\#(\iu\to\is)_{W},
		\end{align*}
		where \(\#(p\to q)_{W}\) denotes the number of times a transition from a vertex \(v_{j}\) to a vertex \(v_{j+1}\) has occurred in \(W\) such that \(\ell_{v_{j}}(i) = p\) and \(\ell_{v_{j+1}}(i) = q\), \(p,q\in\{\is,\iu\}\), \(p\neq q\). The right-hand side of the above equality can be rewritten as
		\begin{align}
		\label{e:pf1step8a}
			m\Biggl(-\abs{\ls}\sum_{\substack{{j:\l_{v_{j}}(i) = \is}\\{j = 0,1,\ldots,n-1}}}T_{v_{j}} &+ \abs{\lu}\sum_{\substack{{j:\l_{v_{j}}(i) = \iu}\\{j = 0,1,\ldots,n-1}}}T_{v_{j}}\nonumber\\+\ln\mu_{\is\iu}\#(\is\to\iu)_{W} &+ \ln\mu_{\iu\is}\#(\iu\to\is)_{W}\Biggr).
		\end{align}
		From the definition of weights associated to vertices and edges of \(G\), we have that the above expression is equal to \(-m\varepsilon_{i}\).
		Also,
		\begin{align}
		\label{e:pf1step8b}
			&-\abs{\ls}D_{s}(mT_{W},t) + \abs{\lu}D_{u}(mT_{W},t)\nonumber\\&+\ln\mu_{\is\iu}N_{\is\iu}(mT_{W},t)+\ln\mu_{\iu\is}N_{\iu\is}(mT_{W},t)\nonumber\\
			&\leq\abs{\lu}(t-mT_{W}) + mn(\ln\mu_{\is\iu}+\ln\mu_{\iu\is}) := a\:\text{(say)}.
		\end{align}
		From \eqref{e:pf1step8a} and \eqref{e:pf1step8b}, we obtain that the right-hand side of \eqref{e:pf1step7} is bounded above by \(\exp\bigl(-m\varepsilon_{i}+a\bigr)\).
		
		Let \(\varphi_{i}:[0,\:t]\to\R\) be a function connecting \((0,\exp(a)+T_{W})\), \((rT_{W},\exp(-(r-1)\varepsilon_{i}+a))\), \((t,\exp(-m\varepsilon_{i}+a))\), \(r = 1,2,\ldots,m\), with straight line segments. By construction, \(\varphi_{i}\) is an upper envelope of \(T\mapsto\psi_{i}(T)\) on \([0{,} t]\), is continuous, decreasing, and tends to \(0\) as \(t\to+\infty\). Hence, \(\varphi_{i}\in\L\).
		
    Since \(i\in\{1,2,\ldots,N\}\) was selected arbitrarily, the assertion of Theorem \ref{t:mainres} follows.
    \end{proof}

    A next natural question is: given the matrices \(A_i\), \(B_i\), \(K_i\), \(i=1,2,\ldots,N\) and the number \(M\), how do we design a \(T\)-contractive cycle \(W = v_0,(v_0,v_1),v_1,\ldots,v_{n-1},(v_{n-1},v_0),v_0\) on the underlying directed graph \(G(V,E)\) of the NCS under consideration? In the remainder of this section we address this question.
    \begin{defn}{\cite[Definition 3]{abc}}
    \label{d:cand_contrac}
    \rm{
        We call a cycle \(W = v_0,(v_0,v_1)\),\(v_1,\ldots,v_{n-1},(v_{n-1},v_0),v_0\) on \(G(V,E)\) \emph{candidate contractive}, if for each \(i=1,2,\ldots,N\), there exists at least one \(v_j\), \(j\in\{0,1,\ldots,n-1\}\) such that \(\ell_{v_{j}}(i) = \is\).
    }
    \end{defn}

    Let \(C_{G}\) denote the set of all candidate contractive cycles on \(G(V,E)\). We next provide an algorithm to design a \(T\)-contractive cycle on \(G(V,E)\).
    \begin{algorithm}[htbp]
	\caption{Design of a \(T\)-contractive cycle on \(G(V,E)\)} \label{algo:cycle_design}
    		\begin{algorithmic}[1]
    			\renewcommand{\algorithmicrequire}{\textbf{Input}:}
			\renewcommand{\algorithmicensure}{\textbf{Output}:}
	
			\REQUIRE The matrices \(A_{i}\), \(B_{i}\), \(K_{i}\), \(i=1,2,\ldots,N\) and the set \(C_G\).
			\ENSURE A \(T\)-contractive cycle \(W = v_0,(v_0,v_1),v_1,\ldots,v_{n-1},(v_{n-1},v_0),v_0\) on \(G(V,E)\) and the corresponding \(T\)-factors \(T_0\), \(T_1,\ldots\), \(T_{n-1}\).
			
			\hspace*{-0.6cm}\textit{Step I: Compute the matrices \(A_{\is}\) and \(A_{\iu}\), \(i=1,2,\ldots,N\)}
			\FOR {\(i=1,2,\ldots,N\)}
				\STATE Set \(A_{\is} = A_{i} + B_{i}K_{i}\) and \(A_{\iu} = A_{i}\)
			\ENDFOR
			
			\hspace*{-0.6cm}\textit{Step II: Fix ranges of values for \(\lambda_p\), \(p\in\{\is,\iu\}\), \(i=1,2,\ldots,N\) and step-sizes \(h_s\) and \(h_u\)}
            \STATE Fix \(\lambda_{s}^{\ell b} > 0\) and \(\lambda_s^{ub} > \lambda_s^{\ell b}\).
            \STATE Fix \(h_s > 0\) (small enough) and compute \(k_s\) such that \(k_s\) is the maximum integer satisfying \(\lambda_s^{\ell b}+k_s h_s\leq \lambda_s^{ub}\).
            \STATE Fix \(\lambda_u^{ub} = 0\) and \(\lambda_u^{\ell b} < \lambda_{u}^{ub}\).
            \STATE Fix \(h_u > 0\) (small enough) and compute \(k_u\) such that \(k_u\) is the largest integer satisfying \(\lambda_u^{\ell b}+k_u h_u \leq \lambda_u^{ub}\).

            \hspace*{-0.6cm}\textit{Step III: Check for pairs \((P_p,\lambda_p)\), \(p\in\{\is,\iu\}\), \(i=1,2,\ldots,N\) under which \(G(V,E)\) admits a \(T\)-contractive cycle}
				
			\FOR {\(\lambda_s = \lambda_s^{\ell b}, \lambda_s^{\ell b}+h_s, \lambda_s^{\ell b} +2h_s,\ldots,\lambda_s^{\ell b}+k_s h_s\)}
                \FOR {\(\lambda_u = \lambda_u^{\ell b}, \lambda_u^{\ell b}+h_u, \lambda_u^{\ell b}+2h_u,\ldots,\lambda_u^{\ell b}+k_u h_u\)}
                    \FOR {\(i=1,2,\ldots,N\)}
					   \STATE Solve for \((P_p,\lambda_p)\), \(p\in\{\is,\iu\}\):
					 	\begin{align}
	           		 	\label{e:feasprob1}
		              				\minimize\:\:&\:\:1\nonumber\\
		              				\sbjto\:\:&\:\:
		                  			\begin{cases}
                              					A_{\is}^\top P_{\is}+P_{\is}A_{\is}\preceq -\ls P_{\is},\\
                              					A_{\iu}^\top P_{\iu}+P_{\iu}A_{\iu}\preceq -\lu P_{\iu},\\
                              					P_{\is}, P_{\iu} \succ 0,\\
                              					\kappa I\preceq P_{\is},P_{\iu}\preceq I_d,\kappa > 0,\\
                                                \kappa\:\text{small enough}.
		                  			\end{cases}
	             		\end{align}
                    \ENDFOR
					 \IF {there is a solution to \eqref{e:feasprob1} for all \(i=1,2,\ldots,N\)}
					 	\STATE Compute \(\mu_{\is\iu} = \lambda_{\max}(P_{\iu}P_{\is}^{-1})\) and \(\mu_{\iu\is} = \lambda_{\max}(P_{\is}P_{\iu}^{-1})\)
                     \ENDIF
						\STATE Solve for \(T_{v_{j}}\), \(j=0,1,\ldots,n-1\):
						 	\begin{align}
	            			\label{e:feasprob2}
		              					\displaystyle{\minimize_{W\in C_G}}\:\:&\:\:1\nonumber\\
		              					\sbjto\:\:&\:\:
		                  				\begin{cases}
                                                    W\:\text{is \(T\)-contractive},\\
                            						T_{v_{j}} > 0,\:\:j=0,1,\ldots,n-1.
                            			\end{cases}
	             			\end{align}
						\IF {there is a solution to \eqref{e:feasprob2}}
					 		\STATE Output \(W = v_0,(v_0,v_1),v_1,\ldots,v_{n-1},(v_{n-1},v_0),v_0\) and the corresponding \(T\)-factors \(T_0\), \(T_1,\ldots\), \(T_{n-1}\), and halt.
						\ENDIF
					\ENDFOR
				\ENDFOR
    		\end{algorithmic}
   \end{algorithm}

   Given the matrices \(A_i\), \(B_i\), \(K_i\), \(i=1,2,\ldots,N\) and the number \(M\), Algorithm \ref{algo:cycle_design} designs a \(T\)-contractive cycle \(W\) on \(G(V,E)\). Recall the computations of the pairs \((P_p,\lambda_p)\), \(p\in\{\is,\iu\}\) from our proof of Fact \ref{fact:key1}. We perform line searches over the intervals \([\ls^{\ell b},\ls^{ub}]\) and \([\lu^{\ell b},\lu^{ub}]\) with step sizes \(h_s\) and \(h_u\), respectively, and solve the feasibility problem \eqref{e:feasprob1} for the pairs \((P_p,\lambda_p)\), \(p\in\{\is,\iu\}\), \(i=1,2,\ldots,N\). The condition \(\kappa I \preceq P_{\is}, P_{\iu}\) limits the condition numbers of \(P_{\is}\) and \(P_{\iu}\) to \(\kappa^{-1}\), and the condition \(P_{\is}, P_{\iu}\preceq I_d\) guarantees that the set of feasible \(P_{\is}, P_{\iu}\) is bounded. If \eqref{e:feasprob1} admits solutions for all \(i=1,2,\ldots,N\), then we compute the scalars \(\mu_{pq}\), \(p,q\in\{\is,\iu\}\), \(i=1,2,\ldots,N\) by using the estimate of \cite[Proposition 4]{chatterjee'14}, and check if any of the candidate contractive cycles on \(G(V,E)\) is \(T\)-contractive. If a solution is found, then Algorithm \ref{algo:cycle_design} outputs the cycle and its corresponding \(T\)-factors, and terminates. Otherwise, the values of \(\ls\) and \(\lu\) are updated, and the search continues.
   \begin{rem}
   \label{rem:compa3}
   \rm{
        The design of \(T\)-contractive cycles on \(G(V,E)\) in this paper differs from \cite[Algorithm 2]{abc} in terms of designing the pairs \((P_p,\lambda_p)\), \(p\in\{\is,\iu\}\), \(i=1,2,\ldots,N\). In \cite[Algorithm 2]{abc} the authors employed designed Lyapunov-like functions in the discrete-time setting, while in the current paper we cater to their continuous-time counterparts.
        }
   \end{rem}
   Algorithm \ref{algo:cycle_design}, however, provides only a partial solution to the problem of designing \(T\)-contractive cycles on the underlying directed graph of an NCS. Indeed, even if the step-sizes \(h_s\) and \(h_u\) are chosen to be very small, only a finite number of possibilities for the pair \((P_p,\lambda_p)\), \(p\in\{\is,\iu\}\), \(i=1,2,\ldots,N\), are explored. Consequently, if no solution to the feasibility problem \eqref{e:feasprob2} is found, we cannot conclude that there does not exist choice of scalars \(\lambda_p\), \(p\in\{\is,\iu\}\) and \(\mu_{pq}\), \(p,q\in\{\is,\iu\}\), \(i=1,2,\ldots,N\), under which \(G(V,E)\) admits a \(T\)-contractive cycle.
   \begin{rem}
   \label{rem:advantage}
   \rm{
        {Switched systems have appeared before in NCSs literature, see e.g., \cite{Ishii'02, Dai'10, Lin'05, Zhang'06}, and average dwell time switching logic is proven to be a useful tool. In the presence of unstable systems, stabilizing average dwell time switching involves two conditions on every interval of time \cite{Liberzon_IOSS}: i) an upper bound on the number of switches and ii) a lower bound on the ratio of durations of activation of stable to unstable subsystems. In contrast, our design of a stabilizing scheduling logic involves design of a \(T\)-contractive cycle on the underlying weighted directed graph of the NCS. To design these cycles, we solve the feasibility problems \eqref{e:feasprob1} and \eqref{e:feasprob2}. We do not impose restrictions on the behaviour of a scheduling logic on every interval of time, thereby leading to numerically tractable stability conditions.}
   }
   \end{rem}
   \begin{rem}
   \label{rem:switched}
   \rm{
        Lyapunov-like functions and graph-theory have been employed to study stability of continuous-time switched systems earlier in the literature, see e.g., \cite{xyz}. In particular, stabilizing switching signals are constructed by concatenating cycles on the underlying directed graph of a switched system that satisfy certain properties. The design of these cycles requires `co'-designing the Lyapunov-like functions and cycles on the directed graph under consideration. This problem is known to be numerically hard. As a natural choice, the existing literature considers the Lyapunov-like functions and a set of scalars corresponding to these functions to be ``given'', and designs cycles on the underlying directed graph of a switched system that satisfy the desired conditions. However, non-existence of such a cycle with the given choice of functions does not conclude non-existence of a class of functions for which the underlying directed graph of a switched system admits a favourable cycle. Algorithm \ref{algo:cycle_design} extends the literature on stability of continuous-time switched systems by providing a partial solution to the problem of `co'-designing Lyapunov-like functions and suitable cycles on the underlying directed graph of a switched system.
   }
   \end{rem}

   We now present two numerical experiments to test our technique for designing stabilizing scheduling logics.
\section{Numerical experiments}
\label{s:num_ex}
\begin{experiment}
\label{ex:exp1}
\rm{
     We consider an NCS with two plants. At any point in time only one plant can communicate with its controller over the shared communication network. More specifically, \(N = 2\) and \(M=1\). The numerical values of the matrices \(A_i\), \(B_i\), \(K_i\), \(i=1,2\) are as follows:
    \begin{align*}
        A_1 &= \pmat{1.2 & 0 & 0 & 0\\0 & 0.8 & 0 & 0\\0 & 0 & 0.4 & 0\\0 & 0 & 0 & 0.2},\\
        B_1 &= \pmat{1 & 0\\0 & 1\\1 & 0\\0 & 1},\\
        K_1 &= \pmat{-40.2184  &  0.0000  & 23.5546  & -0.0000\\
    0.0000 & -34.4621 &  -0.0000 &  18.7252},\\
        \intertext{and}
        A_2 &= \pmat{0.2 & 0 & 0 & 0\\0 & 0.1 & 0 & 0\\0 & 0.05 & 0 & 0\\0 & 0 & 0 & -1},\\
        B_2 &= \pmat{1 & 1\\1 & 1\\1 & 1\\1 & 1},\\
        K_2 &= 1.0e+03 *\pmat{-0.8631  &  1.1950 &  -0.4634  & -0.0102\\
   -0.8631  &  1.1950 &  -0.4634 &  -0.0102}.
    \end{align*}
    The eigenvalues of \(A_1\) and \(A_1+B_1K_1\) are \(0.2, 0.4, 0.8, 1.2\) and \(-14.1705\), \(-0.8933, -14.1542\),\\\(-0.5827\), respectively, while the eigenvalues of \(A_2\) and \(A_2+B_2K_2\) are \(-1, 0.05, 0.1, 0.2\) and \(-282.8432, -0.8686, -0.1699\),\(-0.0777\), respectively.
    We apply our results to design a scheduling logic for the above setting. The following steps are executed:
    \begin{enumerate}[label = \arabic*), leftmargin = *]
        \item The underlying directed graph \(G(V,E)\) of the NCS under consideration has \(2\) vertices and a directed edge between every pair of vertices. The labels of the vertices are:
            \begin{align*}
                L(\overline{v}_1) = \pmat{1_s\\2_u}\:\:\text{and}\:\:L(\overline{v}_2) = \pmat{1_u\\2_s},\:\overline{v}_k\in V,\:k=1,2.
            \end{align*}
        \item We obtain the set of all candidate contractive cycles, \(C_G\) has two components: \(\overline{v}_1,(\overline{v}_1,\overline{v}_2)\),\\\(\overline{v}_2,(\overline{v}_2,\overline{v}_1),\overline{v}_1\) and \(\overline{v}_2,(\overline{v}_2,\overline{v}_1),\overline{v}_1,(\overline{v}_1,\overline{v}_2)\),\(\overline{v}_2\).
        \item We input the matrices \(A_i\), \(B_i\), \(K_i\), \(i=1,2\) and the set of cycles \(C_G\) to Algorithm \ref{algo:cycle_design} and obtain a \(T\)-contractive cycle \(W = \overline{v}_1,(\overline{v}_1,\overline{v}_2),\overline{v}_2,(\overline{v}_2,\overline{v}_1),\overline{v}_1\) with its corresponding \(T\)-factors \(T_{\overline{v}_1} = 19.25\) and \(T_{\overline{v}_2} = 9.36\). The algorithm is implemented in MATLAB R2020a. The following choices of scalars are used: \(\lambda_s^{\ell b} = 0.1\), \(\lambda_u^{ub} = 100\), \(h_s = 0.01\), \(\lambda_u^{\ell b} = 0\), \(\lambda_u^{ub} = 20\), \(h_u = 0.01\), and \(\kappa = 0.01\).
        \item We input the \(T\)-contractive cycle obtained in Step 3) to Algorithm \ref{algo:sched_logic} that designs a scheduling logic \(\gamma\). We have that \(\gamma\) is periodic with period \(T_{\overline{v}_1} + T_{\overline{v}_2} = 28.61\) units of time.
        \item We choose \(10\) different initial conditions from the interval \([-10,10]^{2}\) uniformly at random, and plot \(\norm{x_i(t)}\) versus \(t\) under \(\gamma\) until \(t=150\) units of time, \(i=1,2\). See Figures \ref{fig:xtplot1} and \ref{fig:xtplot2}. GAS of each plant is observed.
    \end{enumerate}
\begin{figure}[!htb]
        \begin{center}
        \includegraphics[scale = 0.6]{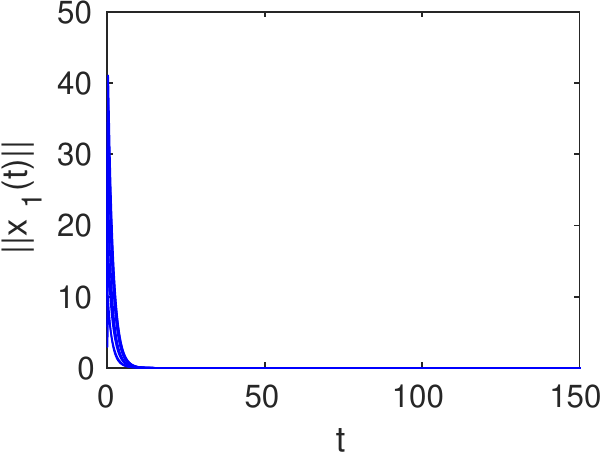}
       \caption{State trajectory for plant 1 under \(\gamma\)}\label{fig:xtplot1}
       \end{center}
      \end{figure}
   \begin{figure}[htbp]
        \begin{center}
        \includegraphics[scale = 0.6]{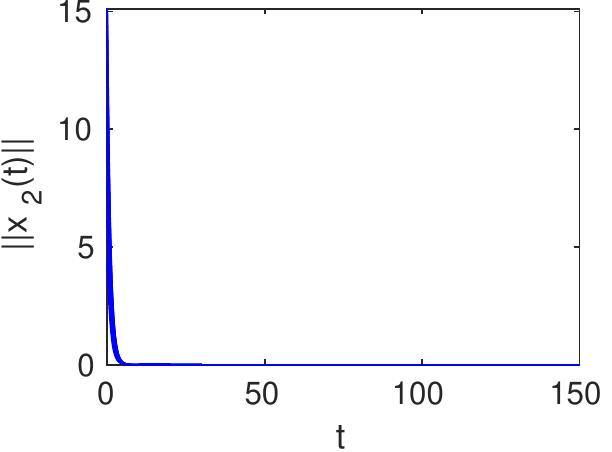}
      \caption{State trajectory for plant 2 under \(\gamma\)}\label{fig:xtplot2}
      \end{center}
\end{figure}
}
\end{experiment}

\begin{experiment}
\label{ex:exp2}
\rm{
	We now test the performance of our techniques in large scale settings, i.e., when the number of plants in the NCS is large. 
        First, we generate \(N\) unstable matrices \(A_{i}\in\R^{2\times 2}\) and vectors \(B_{i}\in\R^{2\times 1}\) with entries from the interval \([-2,2]\) and the set \(\{0,1\}\), respectively, chosen uniformly at random, and ensuring that each pair \((A_{i},B_{i})\), \(i=1,2,\ldots,N\), is controllable. The linear quadratic regulators \(K_{i}\) are computed with \(Q_{i} = Q = 5I_{2\times 2}\) and \(R_{i} = R = 1\). Second, the underlying directed graph \(G(V,E)\) of the NCS is considered, and Algorithm \ref{algo:cycle_design} is employed to design a \(T\)-contractive cycle. Third, this cycle is employed to design stabilizing scheduling logics. In Table \ref{tab:graph_data} we list sizes of \(G\) and lengths of candidate contractive cycles \(W = v_{0},(v_{0},v_{1}),v_{1},\ldots,v_{n-1},(v_{n-1},v_{0}),v_{0}\) for various values of \(N\), and the (rounded-off) computation times. 
    \begin{table}[htbp]
	\centering
	\begin{tabular}{|c | c | c|c|}
		\hline
		\(N\) & \(\abs{V}\) & \(n\) & Computation time (sec)\\
        \hline
        \(100\) & \(1.73\times 10^{13}\) & \(64\) & \(5000\)\\
        \hline
        \(200\) & \(2.24\times 10^{16}\) & \(101\) & \(9000\)\\
        \hline
        \(500\) & \(2.45\times 10^{20}\) & \(327\) & \(37000\)\\
        \hline
        \(700\) & \(7.3\times 10^{21}\) & \(521\) & \(74000\)\\
        \hline
        \(1000\) & \(2.63\times 10^{23}\) & \(803\) & \(90000\)\\
		\hline
	\end{tabular}
    \vspace*{0.2cm}
	\caption{Graph and cycle data}\label{tab:graph_data}
	\end{table}
}
\end{experiment}
\section{Concluding remarks}
\label{s:concln}
    In this paper we addressed the design of scheduling logics for NCSs whose communication networks have limited bandwidth. We presented an algorithm that designs purely time-dependent periodic scheduling logics under which GAS of each plant is preserved. A blend of multiple Lyapunov-like functions and graph theory was employed as the main apparatus for our analysis. The results presented in this paper are a continuous-time counterpart of the techniques proposed in \cite{abc}. 
    
    We identify the following two directions for our future work: First, the extension of our techniques to the design of scheduling logics when the communication networks are also prone to uncertainties like delays, data losses, etc. Second, the co-design of controllers \(K_i\), \(i=1,2,\ldots,N\) and a scheduling logic \(\gamma\) such that good qualitative properties of each plant in an NCS are preserved. The above topics are currently under investigation and will be reported elsewhere.

	

\end{document}